\documentclass{article}
\usepackage{graphicx,multicol}
\usepackage{color, soul}

\usepackage[latin1]{inputenc}

\usepackage{color}
\usepackage{amsthm,amsmath}
\usepackage{amssymb}
\usepackage{amsfonts}
\usepackage{hyperref}
\hypersetup{colorlinks,allcolors=black}
\usepackage{graphicx}
\usepackage{tabularx}
\usepackage{lscape}
\usepackage{pdflscape}

\usepackage{adjustbox} 
\usepackage{array} 

\newtheorem{theorem}{Theorem}
\newtheorem{definition}{Definition}
\newtheorem{corollary}{Corollary}
\newtheorem{prop}{Proposition}
\newtheorem{lemma}{Lemma}
\newtheorem{example}{Example}
\newtheorem{remark}{Remark}

\date{ }

\author{Shakir Ali$^{1}$, Atif Ahmad Khan$^{1,*}$, Abhishek Kesarwani$^2$\\
	\small{$^1$Department of Mathematics, Faculty of Science, }\\
	\small{Aligarh Muslim University, Aligarh 202002, India}\\
	\small{shakir.ali.mm@amu.ac.in, atifkhanalig1997@gmail.com}\\
	\small{$^2$Insight Centre for Data Analytics,  University College Cork,  Cork, Ireland}\\
	\small{akesarwani@ucc.ie}\\ 
	\small{ * Corresponding author email: atifkhanalig1997@gmail.com}}

\begin{document}
\title{On the construction of Cauchy MDS matrices over Galois rings via nilpotent elements and Frobenius maps}
\maketitle
\begin{abstract}

	Let $s,m$ be the positive integers and $p$ be any prime number. Next, let  $GR(p^s,p^{sm})$ be a Galois ring of characteristic $p^s$ and cardinality $p^{sm}$. In the present paper, we explore the construction of Cauchy MDS matrices over Galois rings. Moreover, we introduce a new approach that considers nilpotent elements and  Teichmüller set of Galois ring $GR(p^s,p^{sm})$ to reduce the number of entries in these matrices. Furthermore, we construct $p^{(s-1)m}(p^m-1)$ distinct functions with the help of Frobenius automorphisms. These functions preserve MDS property of matrices. Finally, we prove some results using  automorphisms and isomorphisms of the Galois rings that can be used to generate new Cauchy MDS matrices.


\end{abstract}

\section{Introduction}
	The concept of confusion and diffusion in the design of encryption systems was originally introduced by Claude Shannon in his seminal paper titled ``Communication Theory of Secrecy Systems" \cite{18}. In this context, the primary objective of the confusion layer is to conceal the correlation between the encryption key and the resulting ciphertext, whereas the  diffusion  is a cryptographic technique that would ensure that the effect of one or more then one plaintext digit would be evenly spread out to a number of ciphertext digits. When these principles are applied within an iterated block cipher this guarantees that every individual bit within the message and each secret-key bit exert a nonlinear influence on every bit composing the ciphertext. The branch number of the diffusion matrix being used determines how much small change in the input affects the output of the diffusion layer. It is more difficult for differential and linear attacks to be successful with higher branch number. One effective way to achieve this is by using a Maximum Distance Separable (MDS) matrix, which is known for providing strong diffusion property.
The concept of multipermutation discussed in  \cite{17,21} offers a simple way to represent perfect diffusion formally. Vaudenay in \cite{21} note that a linear multipermutation is same as MDS matrix. MDS matrices are crucial in modern ciphers such as Advanced Encryption Standard (AES) \cite{3}, SQUARE \cite{4}, SHARK \cite{15}, and hash functions \cite{001}. These MDS matrices are also used to create hash functions, playing a central role in hash functions like MAELSTROM-0 \cite{5}, and the PHOTON family of lightweight hash functions \cite{6}. It is worth noting that if $M$ is an MDS matrix used for encryption, then $M^{-1}$ is used for decryption. Consequently, it becomes imperative to select MDS matrices with computationally efficient inverses, especially for software and hardware implementations, as opposed to Feistel-based constructions that do not necessitate inverse transformations.

Certain direct methods for the construction of MDS matrices have already been proposed, involving Cauchy matrices and Vandermonde matrices (see, \cite{2}, \cite{7}, \cite{9}, \cite{10} and \cite{16}, for details). Notably, Youssef et al. \cite{22} introduced a technique to derive involutory MDS matrices from Cauchy matrices, building on the work described in~\cite{10}. Gupta and Ray \cite{7} expanded on this approach by presenting various types of Cauchy matrices suitable for constructing MDS matrices, including involutory Cauchy MDS matrices derived from the additive subgroup of the finite field $F_{2^m}$. Furthermore, Lacan and Fimes \cite{10} introduced a construction of MDS matrices from Vandermonde matrices, and Sajadieh et al.~\cite{16} contributed involutory MDS matrix constructions using Vandermonde matrices. Gupta et al.~\cite{8} later established a connection between Cauchy-based and Vandermonde-based MDS matrix constructions, demonstrating that generalized-Cauchy matrices possess both MDS and involutory properties. In \cite{2}, Cui et al.  gave a construction of higher-level MDS matrices in nested SPNs over finite commutative ring of characteristic 2. The interplay between Cauchy, Vandermonde, and circulant matrices, along with involutory and orthogonal considerations, constitutes a rich research landscape in the field of cryptography, offering valuable insights and solutions for the development of secure cryptographic algorithms.

\textbf{Our contribution:-}
In this paper, we explore the construction of Cauchy Maximum Distance Separable (MDS) matrices over Galois rings. Throughout this  paper, $GR(p^s,~p^{sm})$ will represent the Galois ring of characteristic $p^s$ and cardinality $p^{sm}$ where $m,~s$ are positive integers and $p$ be any prime number.  And $F_{p^m}$, denotes a finite field of characteristic $p$ and cardinality $p^m.$ Let $\mathcal{N}(GR(p^s,~p^{sm}))$  and $U(GR(p^s,~p^{sm}))$ be the set of nilpotent and unit elements of $GR(p^s,~p^{sm})$, respectively. Our approach introduces a novel method for constructing Cauchy matrices by utilizing the Teichmüller set \cite[Chapter 6]{919}. For a Galois ring \(GR(p^s, ~p^{sm})\), the Teichmüller set is \(\tau = \{0, 1, \xi, \xi^2, \dots, \xi^{p^m-2}\}\), where \(o(\xi) = p^m - 1\). The main results of this paper are outlined below. To the best of our knowledge, these results were not known before.

\begin{enumerate}
	\item[(i)] We demonstrate that the matrix \(A = \left[ \frac{1}{x_i - y_j} \right]\) is an MDS matrix of order \(k\) for distinct \(x_i, y_j \in \tau - \{0\}\) with \(1 \leq i, ~j \leq k\).
	\item [(ii)] We show that the matrix \(A = \left[ \frac{1}{x_i + y_j} \right]\) is an MDS matrix of order \(k\) for distinct \(x_i, y_j \in \tau' = \left\{0, 1, \xi, \dots, \xi^{\lceil \frac{p^m-2}{2}\rceil}\right\}\) with \(1 \leq i,~ j \leq k\).
	\item[(iii)] We propose a reduction in the number of elements by involving nilpotent elements, proving that the matrix \(A = \left[ \frac{1}{x_i + x_j + ~l} \right]\) is an MDS matrix of order \(k\) for distinct \(x_i \in \tau' = \left\{0, 1, \xi, \dots, \xi^{\lceil \frac{p^m-2}{2}\rceil}\right\}\) with \(1 \leq i, ~j \leq k\) and \(l \in\mathcal{N}(GR(p^s,~ p^{sm}))\).
\end{enumerate}

Additionally, we investigate distinct functions associated with isomorphisms and automorphisms between Galois rings that preserves MDS property. In particular, we extend the findings of \cite{33}, where the authors identified \(m \cdot (2^m - 1)\) unique functions through isomorphisms expressed as \(f_{su}^c: \beta_1 \mapsto (\beta_2^{su}) \cdot c\). Here, \(\beta_1\) and \(\beta_2\) represent primitive elements in \(F_{2^m}/p_1(x)\) and \(F_{2^m}/p_2(x)\) respectively, with \(c \in F^*_{2^m}\), and \(su = e \cdot 2^i\), where \(1 \leq e \leq 2^m - 2\), \(\gcd(e, 2^m - 1) = 1\), \(p_1(\beta_2^{su}) = 0\), and \(0 \leq u, i \leq m - 1\), where \(F_{2^m}/p_1(x)\) represent a finite field $F_{2^m}$ associated with the polynomial $p_1(x)$. These functions facilitate the generation of  new MDS matrices over \(F_{2^m}/p_2(x)\) from an existing MDS matrix over \(F_{2^m}/p_1(x)\), where \(p_1(x)\) and \(p_2(x)\) are irreducible polynomials over \(F_2\).

Our findings significantly enhance the understanding and construction of MDS matrices over Galois rings, with potential applications in cryptography  and related fields.

\textbf{Organization of paper:-} This paper is structured as follows:  In Section 2, we provide some basic  definitions and results that we use in later section. In Section 3, we present our main results on the construction of Cauchy MDS matrices over $GR(p^s,~p^{sm}).$  Further, we discuss this construction for Cauchy matrices of TYPE-I and TYPE-II in more details. In Section 4, we provide distinct functions which maps MDS matrices to MDS matrices.
In Section 5, we justify our results with illustrative examples of Cauchy matrices. Finally, we conclude this paper in Section 6.

\section{Preliminaries}\label{sec_pre}
	In this section, we define key notions, some well-known facts about Galois rings and state some important results. We begin our discussions with the following definition:
\begin{definition}\cite[Definition 3.2]{111}
	A finite commutative ring with unity $R$ such that the set of its zero divisors including $0$ constitutes a principal ideal $( p )$ with prime $p$  (i.e., $R/( p )$ is an integral domain) is called a Galois ring.
\end{definition}
\begin{example}
	The rings $\mathbb{Z}_4$, $\mathbb{Z}_8$, and $\mathbb{Z}_9$ are finite (commutative) rings with unity for which the set of zero divisors (including $0$) constitutes a principal ideal $( p = 2 )$, $(p = 2 )$, and $( p = 3 )$, respectively. Therefore, $\mathbb{Z}_4 (= \mathbb{Z}_{2^2})$, $\mathbb{Z}_8 (= \mathbb{Z}_{2^3})$, and $\mathbb{Z}_9 (= \mathbb{Z}_{3^2})$ are Galois rings. Note that $( p = 2 )$, $( p = 2 )$, and $( p = 3 )$ are the unique maximal ideals of $\mathbb{Z}_{2^2}$, $\mathbb{Z}_{2^3}$, and $\mathbb{Z}_{3^2}$, respectively.
\end{example}

Let $p$ be a fixed prime and $s$ be a positive integer. First, we consider the following canonical projection
\begin{equation}
	\mu : \mathbb{Z}_{p^s} \rightarrow \mathbb{Z}_p
\end{equation}
which is defined by
\begin{equation}
	\mu(c) = c \ (\text{mod}\ p).
\end{equation}
The map $\mu$ can be extended naturally to the following map
\begin{equation}
	\mu : \mathbb{Z}_{p^s}[x] \rightarrow \mathbb{Z}_p[x]
\end{equation}
which is defined by
\begin{equation}
	\mu(a_0 + a_1x + \dots + a_nx^n) = \mu(a_0) + \mu(a_1)x + \dots + \mu(a_n)x^n.
\end{equation}
This extended $\mu$ is a ring homomorphism with kernel $(p)$.
Let $f(x)$ be a polynomial in $\mathbb{Z}_{p^s}[x]$. Then, $f(x)$ is called basic irreducible if $\mu(f(x))$ is irreducible in $\mathbb{Z}_p[x]$. A Galois ring is constructed as
\begin{equation}
	GR(p^s,~p^{sm}) = \mathbb{Z}_{p^s}[x]/(f(x)),
\end{equation}
where $f(x)$ is a monic basic irreducible polynomial in $\mathbb{Z}_{p^s}[x]$ of degree $m$. The elements of $GR(p^s,p^{sm})$ are residue classes of the form
\begin{equation}
	a_0 + a_1x + \dots + a_{m-1}x^{m-1} + (f(x)),
\end{equation}
where $a_i \in \mathbb{Z}_{p^s}$, $(0 \leq i \leq m - 1)$. The ring homomorphism $\mu$ induces a ring homomorphism $\tilde{\mu}$
\begin{equation}\label{eqn21}
	\tilde{\mu} : GR(p^s,~p^{sm}) = \mathbb{Z}_{p^s}[x]/(f(x)) \rightarrow F_{p^m} = \mathbb{Z}_p[x]/(\mu(f(x)))
\end{equation}
which is defined by
\begin{equation}
	\tilde{\mu}(g(x) + (f(x))) = \mu(g(x)) + (\mu(f(x))),~where~ g(x) \in \mathbb{Z}_{p^s}[x].
\end{equation}
A polynomial $h(x)$ in $\mathbb{Z}_{p^s}[x]$ is called a basic primitive polynomial if $\mu(h(x))$ is a primitive polynomial in $\mathbb{Z}_p[x]$. It is a well-known fact that there is a monic basic primitive polynomial $h(x)$ of degree $m$ over $\mathbb{Z}_{p^s}$ and $h(x)|(x^{p^m-1} - 1)$ in $\mathbb{Z}_{p^s}[x]$. Let $h(x)$ be a monic basic primitive polynomial in $\mathbb{Z}_{p^s}[x]$ of degree $m$. Consider the following element
\begin{equation}\label{eqn22}
	\xi = x + (h(x)) \in GR(p^s,~p^{sm}) = \mathbb{Z}_{p^s}[x]/(h(x)).
\end{equation}
Then, the order of $\xi$ is $p^m - 1$. Teichmüller representatives are defined as follows:
\begin{equation}
	\tau = \{0, 1, \xi, \xi^2, \ldots, \xi^{p^m-2}\}.
\end{equation}
Also, every element $t \in GR(p^s,~p^{sm})$ can be uniquely represented in the form
\begin{equation}
	t = t_0 + pt_1 + p^2t_2 + \cdots + p^{s-1}t_{s-1},
\end{equation}
where $t_i \in T$, $(0 \leq i \leq s - 1)$. Using this notation, the following map $\sigma$ define as
\begin{equation}
	\sigma : GR(p^s,~p^{sm}) \rightarrow GR(p^s,~p^{sm})
\end{equation}
by
\begin{equation}
	\sigma(t) = t_0^p + pt_1^p + p^2t_2^p + \cdots + p^{s-1}t_{s-1}^p.
\end{equation}
The following facts are well known about the map $\sigma$ defined above:
\begin{enumerate}
	\item $\sigma$ is a ring automorphism of $GR(p^s,~p^{sm})$.
	\item $\sigma$ fixes every element of $\mathbb{Z}_{p^s}$.
	\item $\sigma$ is of order $m$ and generates the cyclic Galois group of $GR(p^s,~p^{sm})$ over $\mathbb{Z}_{p^s}$.
\end{enumerate}
For further studies on Galois rings, we refer readers to see the reference \cite{11}. 
\begin{example}
	Consider the ring $\mathbb{Z}_{p^s}$, where $p$ is a prime number and $s$ is a positive integer. Clearly, 1 is the identity of $\mathbb{Z}_{p^s}$ and the set of its zero divisors including 0 forms a maximal ideal $(p)$. Hence, $\mathbb{Z}_{p^s}$ is a Galois ring with $p^s$ elements.
\end{example}

\begin{definition}\cite[Definition 6]{20}
	Let $GR(p^s,~p^{sm})$ be a Galois ring and  $M$ be a  matrix of order $k$ over $GR(p^s,~p^{sm})$. Then, $M$ is an MDS matrix if every square submatrices of $M$ are non-singular. 
\end{definition}
\begin{definition}\label{def1}
	Let $R$ denotes a finite commutative ring with unity and $U(R)$ be the set of units of $R$. Then, a matrix $A$  defined as
	\begin{eqnarray*}
		A&=&\Big[\frac{1}{x_{i}-y_{j}}\Big],
	\end{eqnarray*}
	is said to be Cauchy matrix of the first kind of order $k$ if for any $x_{i},~y_{j}\in U(R)$ with $1\leq i,j\leq k,$ and the following conditions are satisfied: 
	\begin{enumerate}
		\item For $1\leq i\neq j\leq k$, $x_{i}-x_{j},~y_{j}-y_{i} \in U(R).$
		\item For $1\leq i,~ j\leq k$, $x_{i}-y_{j}\in U(R).$
	\end{enumerate}
	Moreover, the determinant of $A$ is given by,
	\begin{eqnarray*}
		det(A)=\frac{\prod_{i=2}^{k}\prod_{j=1}^{i-1}(x_{i}-x_{j})(y_{j}-y_{i})}{\prod_{i=1}^{k}\prod_{j=1}^{k}(x_{i}-y_{j})}.
	\end{eqnarray*}
	
\end{definition}
\begin{remark}\cite{7}\label{r1}
	It is straightforward to check that every submatrix of a Cauchy matrix is a Cauchy matrix.
\end{remark}
\begin{definition}
	Let $R$ be a finite commutative ring with unity. Then, a matrix $A$ defined as
	\begin{eqnarray*}
		A&=&\Big[\frac{1}{x_{i}+y_{j}}\Big],
	\end{eqnarray*}
	is said to be Cauchy matrix of the second kind of order $k$ if for any $x_{i},~y_{j}\in U(R)$  with $1\leq i,j\leq k$ and the following conditions are satisfied: 
	\begin{enumerate}
		\item For $1\leq i\neq j\leq k$, $x_{i}-x_{j},~y_{j}-y_{i} \in U(R).$
		\item For $1\leq i,~ j\leq k$, $x_{i}+y_{j}\in U(R).$
	\end{enumerate} Also, the determinant of this matrix is given by,
	\begin{eqnarray*}
		det(A)=\frac{\prod_{i=2}^{k}\prod_{j=1}^{i-1}(x_{i}-x_{j})(y_{i}-y_{j})}{\prod_{i=1}^{k}\prod_{j=1}^{k}(x_{i}+y_{j})}.
	\end{eqnarray*}
	
\end{definition}
\begin{definition}\label{d11}
	\textbf{(Generalized Cauchy matrix)} 	Let $R$ be a finite commutative ring with unity and  $x_{i},~y_{j},~ u_{i},~v_{j} \in U(R)~(0\leq i,j \leq k-1)$ such that $x_{i}-y_{j},~x_{i}-x_{j},$ and $y_{j}-y_{i} \in U(R)$. Then, a $k\times k$ matrix $A=(a_{i,j})$ with
	\begin{eqnarray*}
		a_{i,j}=\frac{u_{i}v_{j}}{x_{i}-y_{j}},
	\end{eqnarray*}
	is called a Generalized Cauchy matrix (GC matrix).
	\begin{remark}
		If we take $u_i=v_j=1 $ for $1\leq i,~j\leq k-1$ in the Definition \ref{d11}, then we get Cauchy matrix defined in Definition \ref{def1}.
	\end{remark}
\end{definition}
\begin{definition}
	A square matrix B is called an involutory matrix if  $B^2=I,~i.e., ~B=B^{-1}.$
\end{definition}
\begin{remark}
	A square matrix over a ring is non-singular if and only if its determinant is a unit.
\end{remark}
\begin{lemma}\cite[Lemma 7]{20}
	Let $\mu$ be the epimorphism from $GR(p^s,~p^{sm})$ to $F_{p^m}$ and $M=(a_{ij})$ be a $k\times k$ matrix over $U(GR(p^s,~p^{sm}))$, where $a_{ij}\in U(GR(p^s,~p^{sm}))$. Suppose $\overline{M}=(\mu(a_{ij}))$ over $F_{p^m}$. Then, $\mu(det(M))=\det(\overline{M})$, where $\det(M)$ is the determinant of $M$.
\end{lemma}
\begin{theorem}\cite[Theorem 1]{20}\label{t5}
	Let $\overline{\mu}$ be the epimorphism from $GR(p^s,~p^{sm})$ to $F_{p^m}$, and let $M = (a_{ij})$ be a $k \times k$ matrix over $U(GR(p^s,~p^{sm}))$, where $a_{ij} \in U(GR(p^s,~p^{sm}))$. The matrix $M$ is an MDS matrix if and only if $\overline{M} = (\mu(a_{ij}))$ over $F_{p^m}$ is an MDS matrix of order $k$.
	
\end{theorem}

\section{\textbf{The main results}}
In \cite{7}, the authors provided the construction of Cauchy MDS matrices over finite fields. In the present section, we  construct Cauchy MDS matrices over a Galois ring $GR(p^s,~p^{sm})$. We begin our discussions with the first result of this paper.
\begin{theorem}\label{t6}
	Let  $GR(p^s,~p^{sm})$ be a Galois ring and $\xi \in GR(p^s,~p^{sm}) $ such that o$(\xi)=p^{m}-1$. Next, let  $\tau=\{0,1,\xi,\xi^2,\dots, \xi^{p^m-2}\}$ be a Teichmüller set. Then, 
	for any distinct elements $x_{1},~x_{2},\dots,~x_{k},~ y_{1},~ y_{2},~\dots,~ y_{k}\in \tau-\{0\}$, the Cauchy matrix $A=\big[\frac{1}{x_{i}-y_{j}}\big]$ is an MDS matrix of order $k$.
	
\end{theorem}
\begin{proof}
	We are	given a set  $\tau=\{0,1,\xi,\xi^2,\dots, \xi^{p^m-2}\}$ and  $x_{i} \in \tau-\{0\}$, for $1\leq i \leq k$ and $y_{j}\in \tau-\{0,x_{1},x_{2},\dots,x_{k}\}$, for $1\leq j \leq k$ are distinct elements in  $\tau$. Then matrix $A$ is defined as 
	\begin{eqnarray*}
		A&=&\Bigg[\frac{1}{x_{i}-y_{j}}\Bigg],
	\end{eqnarray*}
	and its determinant is given by
	\begin{equation}\label{e1}
		\det(A)=\frac{\prod_{i=2}^{n}\prod_{j=1}^{i-1}(x_{i}-x_{j})(y_{j}-y_{i})}{\prod_{i=1}^{n}\prod_{j=1}^{n}(x_{i}-y_{j})}.
	\end{equation}
	Since $x_{i},y_{j}\in \tau-\{0\}$ and $o(\xi)=p^m-1$, so we have
	$$\xi^{p^m-1}=1.$$
	For $i=0,1,\dots,p^m-2$, $\xi^i$ are units in $GR(p^s,~p^{sm}) $ such that
	$$\xi^i\cdot\xi^{p^m-1-i}=1.$$
	Moreover, all elements of the form $1-\xi^j$ where $0<j<p^m-2$, are also units in $GR(p^s,~p^{sm})$. Let on contrary that $1-\xi^j$ belongs to the maximal ideal $(p)$ of $GR(p^s,~p^{sm})$. Then,  by Equation (\ref{eqn21}), we have $\bar{\xi}^j=1$, which contradicts the fact that the order of $\bar{\xi}$ is $p^m-1$ in $F_{p^m}$. It follows that, for $0\leq i < j\leq p^m-2$, all $\xi^i-\xi^j$ are units in $GR(p^s,~p^{sm})$. Consequently, $x_{j}-x_{i}$, $y_{j}-y_{i}$, and $x_{i}-x_{j}$ are also units in $GR(p^s,~p^{sm})$. Therefore, by Equation (\ref{e1}), we obtain, $\det(A)\in U(GR(p^s,~p^{sm}))$. Then, by Remark \ref{r1}, all of its submatrices are Cauchy matrices. Hence, $A$ is an MDS matrix.
\end{proof}

The construction mentioned above is known as the Cauchy construction of TYPE-I over a Galois ring.
\begin{remark}
	Theorem \ref{t6} is not true for a Cauchy matrix of the second kind, as demonstrated by the following example.
\end{remark}
\begin{example}
	Let $GR(3^2,~(3^2)^2)=\frac{\mathbb{Z}_{3^2}[x]}{(5x^2+2x+4)}$ be Galois ring of characteristic 9 with cardinality 81 and $\xi=x+(5x^2+2x+4)$ such that $o(\xi)=8$. Define $\tau= \{0,1,\xi,\xi^2,\dots, \xi^{7}\}$, and set
	\begin{eqnarray*}
		x_{1}=1 &\space ;& y_{1}=\xi^3,\\
		x_{2}=\xi &\space ;& y_{2}=\xi^4,\\
		x_{3}=\xi^2 &\space ;& y_{3}=\xi^5.
	\end{eqnarray*}
	Since $x_{2}+y_{3}=\xi+\xi^5=\xi-\xi=0$, we can not define Cauchy matrix over the above mentioned entries. Hence, Theorem \ref{t6} is not true for a Cauchy matrix of the second kind.
\end{example}
In the following theorems, we  establish the conditions on the Teichmüller set $\tau$ that allow the existence of the Cauchy matrix of the second kind over Galois rings.
\begin{theorem}\label{t17}
	Let $GR(p^s,~p^{sm})$
	be a Galois ring of characteristic $p^s(p\neq 2)$  and $\xi \in GR(p^s,~p^{sm}) $ such that o$(\xi)=p^{m}-1$. Next, let   $\tau^{'}=\{0,1,\xi,\xi^2,\dots,\xi^{\lceil \frac{p^m-2}{2}\rceil}\}$ be a set, where $\lceil\cdot \rceil$ denotes the greatest integer function. Then, the following statements hold:
	\begin{enumerate}
		\item For distinct $x_{1},x_{2},\dots,x_{k} \in \tau^{'}-\{0\}$ and distinct $y_{1},y_{2},\dots,y_{k} \in \tau^{'}-\{0,x_{1},x_{2},\dots,x_{k}\}$, the Cauchy matrix $A=\big[\frac{1}{x_{i}+y_{j}}\big]$ is an MDS matrix of order $k$.

		\item  For distinct $x_{1},x_{2},\dots,x_{k} \in \tau^{'}-\{0\}$ and $l \in \mathcal{N}(GR(p^s,p^{sm}))$, define $y_{j}=x_{j}+l$,
		the Cauchy matrix $A=\big[\frac{1}{x_{i}+x_{j}+l}\big]$ is an MDS matrix of order $k$.
	\end{enumerate}
\end{theorem}

\begin{proof}(a)
	Given that distinct $x_{i},~y_{j} \in \tau'-\{0\}$, for $1\leq i,j \leq k$. Then, we have
	\begin{eqnarray*}
		A&=&\Bigg[\frac{1}{x_{i}+y_{j}}\Bigg],
	\end{eqnarray*}
	and
	\begin{equation}\label{e2}
		\det(A)=\frac{\prod_{i=2}^{n}\prod_{j=1}^{i-1}(x_{i}-x_{j})(y_{j}-y_{i})}{\prod_{i=1}^{n}\prod_{j=1}^{n}(x_{i}+y_{j})}.
	\end{equation}
	
	Since $x_{i},~y_{j}\in \tau'
	-\{0\}$ and $o(\xi)=p^m-1,$ so
	$\xi^{p^m-1}=1.$ Also,	for $i=0,1,\dots,p^m-2$, $\xi^i$ are units in $GR(p^s,~p^{sm}) $ as 
	$\xi^i\cdot\xi^{p^m-1-i}=1.$
	For any  integer $0<j<\lceil \frac{p^m-2}{2}\rceil$,  we want to check that whether $1+\xi^j$ are   units or nilpotents in $GR(p^s,~p^{sm})$. For this, let us assume on the contrary that  $1+\xi^j \in (p)$. Then by Equation (\ref{eqn21}), we have
	\begin{eqnarray*}
		\bar{\xi}^j&=&-1,\\
		\bar{\xi}^{j}&=&p^m-1,\\
		\bar{\xi}^{2j}&=&(p^m-1)^2,\\
		\bar{\xi}^{2j}&=&1,
	\end{eqnarray*} where $0<2j<p^m-2$, which contradicts the order of $\bar{{\xi}}$ being $p^m-1$. It follows that for $0\leq i < j\leq \lceil \frac{p^m-2}{2}\rceil$, all $\xi^i+\xi^j$ are units in $GR(p^s,~p^{sm})$. Consequently, $x_{j}-x_{i},~ y_{j}-y_{i}$, and $x_{i}+x_{j}$ are units in $GR(p^s,~p^{sm})$. From Equation (\ref{e2}), we conclude that $\det(A)\in U(GR(p^s,~p^{sm}))$. Therefore, by Remark \ref{r1}, all of its submatrices are Cauchy matrices. Hence, $A$ is MDS matrix.
\end{proof}
\begin{proof} (b)
	Given that $x_{i}\in \tau'-\{0\}$, for $1\leq i \leq k$ and $l$ is any fixed nilpotent element of $GR(p^s,~p^{sm})$, for $1\leq j \leq k$, we  define $$y_{j}=x_{j}+l,$$ then the Cauchy matrix
	\begin{eqnarray*}
		A&=&\Bigg[\frac{1}{x_{i}+y_{j}}\Bigg]\\
		&=&\Bigg[\frac{1}{x_{i}+x_{j}+l}\Bigg].
	\end{eqnarray*} 
	The determinant of $A$ is defined as
	\begin{eqnarray*}
		\det(A)&=&\frac{\prod_{i=2}^{n}\prod_{j=1}^{i-1}(x_{i}-x_{j})(x_{j}-l-x_{i}+l)}{\prod_{i=1}^{n}\prod_{j=1}^{n}(x_{i}+y_{j}+l)}\\
		&=&\frac{\prod_{i=2}^{n}\prod_{j=1}^{i-1}(x_{i}-x_{j})(x_{j}-x_{i})}{\prod_{i=1}^{n}\prod_{j=1}^{n}(x_{i}+x_{j}+l)}.
	\end{eqnarray*}
	Application of Theorem \ref{t6} yields, $x_{i}-x_{j},~ y_{j}-y_{i} \in U(GR(p^s,~p^{sm})).$
	To show that $\det(A)$ is a unit element in $GR(p^s,~p^{sm})$, first we prove that $x_{i}+x_{j}+l \in U(GR(p^s,~p^{sm}))$. We consider the following two cases:
	\begin{enumerate}
		\item \textbf{Case-I}:  For $i=j$, $x_{i}+x_{j}+l=x_{i}+x_{i}+l=2x_{i}+l$ as 2 is unit in $GR(p^s,p^{sm})$, this gives $2x_{i}+l$ is a unit element in $U(GR(p^s,~p^{sm})).$ 
		\item \textbf{Case-II}:  For $i\neq j$, we have $x_{i}+x_{j}+l$ is a unit in  $GR(p^s,~p^{sm})$ by part (a).
	\end{enumerate}
	Therefore, $\det(A)\in U(GR(p^s,~p^{sm})).$ Hence, by Remark \ref{r1} determinant of every submatrix of matrix $A$ is invertible. This shows that $A$ is an MDS matrix.
\end{proof}
The construction in Theorem \ref{t17}$(b)$ is known as Cauchy TYPE-II construction over  Galois rings.
In this construction, we have reduced the number of distinct entries in the Cauchy matrix. In the previous construction (Theorem \ref{t6}), referred to as Cauchy TYPE-I, the number of distinct entries were at most $k^2$. However, in Cauchy TYPE-II discussed above, we have reduced the number of distinct entries to at most $\frac{k(k+1)}{2}$. 
\begin{theorem}\label{t8}
	Let $GR(2^s,~2^{sm})$ be a Galois ring of characteristic $2^s\neq 2$  and $\xi \in GR(2^s,~2^{sm}) $ such that o$(\xi)=2^{m}-1$. Next, let  $\tau=\{0,1,\xi,\xi^2,\dots, \xi^{2^m-2}\}$ be a set. For  distinct $x_{i},~y_{j} \in \tau-\{0\}$, $1\leq i,j\leq k$, the Cauchy matrix 
	\begin{eqnarray*}
		A=\Big[\frac{1}{x_{i}+y_{j}}\Big],
	\end{eqnarray*}
	is an MDS matrix of order $k$.
\end{theorem}
\begin{proof}
	Given that $GR(2^s,~2^{sm})$ is a Galois ring of characteristic equal to $2^s$ and cardinality $2^{sm}$. Then, in view of relation (\ref{eqn22}), there exists $\xi \in GR(2^s,~2^{sm})$ such that $o(\xi)=2^m-1$ and 
	$\xi^{2^m-1}=1.$
	For any $\xi^i$, $\xi^j$ in $\tau$, $\xi^i-\xi^j$ should be unit, for $0\leq i<j\leq 2^m-2.$
	Since $2^m-1$ is odd, so
	\begin{eqnarray*}
		\xi^i+\xi^j \in U(GR(2^s,2^{sm})),~~\textup{for}~ 0\leq i<j\leq 2^m-2.
	\end{eqnarray*}
	For distinct $x_{i},~y_{j} \in \tau-\{0\}$, $1\leq i,~j\leq k$, the determinant of matrix
	\begin{eqnarray*}
		A&=&(a_{i,j})=\Big[\frac{1}{x_{i}+y_{j}}\Big],
	\end{eqnarray*}
	is defined as,
	\begin{eqnarray*}
		\det(A)=\frac{\prod_{i=2}^{n}\prod_{j=1}^{i-1}(x_{i}-x_{j})(y_{j}-y_{i})}{\prod_{i=1}^{n}\prod_{j=1}^{n}(x_{i}+y_{j})}.
	\end{eqnarray*}
	Since $x_{i}-x_{j},~y_{i}-y_{j}$ and $x_{i}+y_{j} \in U(GR(2^s,2^{sm}))$, we conclude that $A$ is an MDS matrix.
\end{proof}
\begin{theorem}\label{t9}
	Let $GR(p^s,~p^{sm})$ be a Galois ring of characteristic $p^s(p\neq 2)$ and $\xi \in GR(p^s,~p^{sm}) $ such that o$(\xi)=p^{m}-1$. Next, let  $\tau=\{0,1,\xi,\xi^2,\dots, \xi^{p^m-2}\}$ be a set. For a distinct $x_{i}=\xi^{\sigma_{i}},~y_{j}=\xi^{\eta_{j}} \in \tau-\{0\}$, $1\leq \sigma_{i},~\eta_{j} \leq p^m-2$ such that $\sigma_{i}-\eta_{j}\neq \frac{p^m-1}{2}$. Then, the Cauchy matrix defined by
	\begin{eqnarray*}
		A=\Bigg[\frac{1}{x_{i}+y_{j}}\Bigg],
	\end{eqnarray*}
	is an MDS matrix of order $k$.
\end{theorem}
\begin{proof}
	Given that $GR(p^s,p^{sm})$ is a Galois ring of odd characteristic, for $1\leq i<j\leq p^m-2$ $\xi^i-\xi^j$ is a element of $U(GR(p^s,p^{sm}))$. Therefore, for $x_{i}=\xi^{\sigma_{i}}$, $\xi^{\sigma_{i}}-\xi^{\sigma_{j}}$ is a unit whenever $\sigma_{i}\neq\sigma_{j}$. This gives $\xi^{\sigma_{i}}+\xi^{\sigma_{j}}$ is not unit when $\xi^{\sigma_{i}-\sigma_{j}}=-1~(\sigma_{i}\geq\sigma_{j})$. This implies
	$ \xi^{2(\sigma_{i}-\sigma_{j})}=1$ and hence $ (p^m-1) |2(\sigma_{i}-\sigma_{j})$, that is, 
	$ 2(\sigma_{i}-\sigma_{j})=k(p^m-1)$.
	Moreover, the equality holds only when $k=1$, i.e.,
	$$\sigma_{i}-\sigma_{j}=\frac{p^m-1}{2}.$$
	Consequently, Cauchy matrix $A$ is an MDS matrix whenever $\sigma_{i}-\sigma_{j}\neq\frac{p^m-1}{2}.$
\end{proof}
In the next result, we investigate MDS matrices over the extension of Galois rings via Frobenius automorphisms. One can generate new MDS matrices over Galois ring  extension by applying Frobenius automorphisms.
\begin{theorem}\label{t10}
	Let $l>1$ be any positive integer and $GR(p^s,~p^{sml})$ be a Galois ring. For $1\leq t\leq l-1$, define   automorphisms $\phi^t$ as follows:
	\begin{eqnarray*}
		\phi^{t}:GR(p^s,~p^{sml}) &\rightarrow& GR(p^s,~p^{sml})\\
		\phi^{t}(a_{0}+a_{1}\xi+\dots+a_{l-1}\xi^{l-1})&=&a_{0}+a_{1}\xi^{p^{m\cdot t}}+\dots+a_{l-1}\xi^{(l-1)p^{m\cdot t}},
	\end{eqnarray*}
	where $a_{i}\in GR(p^s,p^{sm})$.  For distinct $x_{i},~y_{j} \in \tau-\{0\}$, $1\leq i,~j\leq k$, the matrix defined by
	\begin{eqnarray*}
		A=\Big[\frac{1}{\phi^{t}(x_{i}+y_{j})}\Big],
	\end{eqnarray*}
	is an MDS matrix.
\end{theorem}

\begin{proof}
	Given that $GR(p^s,~p^{sml})$ is a Galois ring.  Then by \cite[Theorem 14.30]{11}, $GR(p^s,~p^{sml})$ is an extension ring of $GR(p^s,~p^{sm})$. By Equation (\ref{eqn22}), there exists $\xi \in GR(p^s,~p^{sml})$ such that
	\begin{eqnarray*}
		o(\xi)&=&p^{ml}-1,~~and\\
		GR(p^s,~p^{sml})&=&\{a_{0}+a_{1}\xi+a_{2}\xi^2+\dots+a_{l-1}\xi^{l-1};~ a_{i}\in GR(p^s,~p^{sm})\},
	\end{eqnarray*}
	and by \cite[Theorem 14.30]{11} define automorphisms as,
	\begin{eqnarray*}
		\phi^t:GR(p^s,~p^{sml})&\rightarrow& GR(p^s,~p^{sml})\\
		\phi^{t}(a_{0}+a_{1}\xi+\dots+a_{l-1}\xi^{l-1})&=&a_{0}+a_{1}\xi^{p^{m\cdot t}}+\dots+a_{l-1}\xi^{(l-1)p^{m\cdot t}}.
	\end{eqnarray*}
	If we take $\xi^i \in \tau$ for some $i$, then $\phi(\xi^i)=\xi^{ip^m} \in \tau-\{0\}$. Thus, we have  
	\begin{eqnarray*}
		\det(A)=\frac{\prod_{i=2}^{n}\prod_{j=1}^{i-1}(\phi(x_{i}-x_{j}))\phi((y_{j}-y_{i}))}{\prod_{i=1}^{n}\prod_{j=1}^{n}\phi((x_{i}+y_{j}))}.
	\end{eqnarray*}
	Since $x_{i}-x_{j},~y_{j}-y_{i},~and~x_{i}+y_{j}$ are units in $U(GR(p^s,~p^{sml}))$, so images of those elements must be unit in $GR(p^s,~p^{sml})$. Thus, we obtain $ \det A\in U(GR(p^s,~p^{sml})).$ Hence, 
	$A$ is invertible.
	By  Remark \ref{r1}, every submatrix of $A$ is invertible and this gives $A$ is a MDS matrix.
\end{proof}
\begin{remark}
	For each $\phi^t$ $(1\leq t \leq l-1),$ by Theorem \ref{t10}, we can see,  that $A$ is MDS if and only if $A^{\phi^t}$ is MDS. Thus, for given  Cauchy matrix $A$, we can generate up to $l$ many Cauchy matrices. 
\end{remark}
\begin{theorem}
	Let $p$ be a odd prime and $GR(p^s,~p^{sm})$ be a Galois ring of characteristic $p^s$ with  cardinality $p^{sm}$ and $\xi \in GR(p^s,~p^{sm}) $ such that o$(\xi)=p^{m}-1$. Next, let  $\tau=\{0,1,\xi,\xi^2,\dots, \xi^{p^m-2}\}$ be a set. For distinct $x_{i},~y_{j} \in \tau-\{0\}$ and $w_{i},~v_{j}\in U(GR(p^s,p^{sm}))$, $1\leq i,~j\leq k$  be a set of units. Then, the  matrix $A=(a_{i,j})$ with
	\begin{eqnarray*}
		a_{i,j}=\frac{w_{i}v_{j}}{x_{i}-y_{j}},
	\end{eqnarray*}
	
	is an MDS matrix of order $k$.
\end{theorem}
\begin{proof}
	The determinant of matrix $A$ is defined as:
	\begin{eqnarray*}
		\det(A)=w_{1}w_{2}\cdots w_{k}v_{1}v_{2}\cdots v_{k}\frac{\prod_{i=2}^{n}\prod_{j=1}^{i-1}(x_{i}-x_{j})(y_{j}-y_{i})}{\prod_{i=1}^{n}\prod_{j=1}^{n}(x_{i}+y_{j})}.
	\end{eqnarray*}
	Since $x_{i}-x_{j},~y_{i}-y_{j},$ and $x_{i}+y_{j}$ are in $U(GR(p^s,~p^{sm}))$ and products of units are unit, so $\det(A)$ is unit in $U(GR(p^s,~p^{sm})).$ Hence, $A$ is an MDS Cauchy matrix.
\end{proof}
Since we know that involutory matrices, known for their self-inverse and play a pivotal role in the design of block ciphers. However, in our next theorem, we demonstrate that it is impossible for a Cauchy TYPE-II construction to yield an MDS matrix that is also involutory.
\begin{theorem}
	Let $GR(p^s,~p^{sm})$ be a Galois ring of characteristic $p^s$(where \( p \neq 2 \)) and cardinality \( p^{sm} \). Then, there is no Cauchy TYPE-II construction that yields an almost involutory MDS matrix of order 2 over a Galois ring \( GR(p^s,~p^{sm}) \).
	
\end{theorem}
\begin{proof}
	Let $\xi \in \tau$ such that $\xi^{p^m-1}=1$. Suppose $x_{1}, x_{2} \in \tau-\{0\}$, define $y_{1}=x_{1}+l,~y_{2}=x_{2}+l;~~l\in \mathcal{N}(GR(p^s,p^{sm})).$ Then, we have
	\begin{eqnarray*}
		A&=& \begin{bmatrix}
			\frac{1}{x_{1}+y_{1}}&\frac{1}{x_{1}+y_{2}}\\
			\frac{1}{x_{2}+y_{1}}&\frac{1}{x_{2}+y_{2}}
		\end{bmatrix}\\
		&=&\begin{bmatrix}
			\frac{1}{x_{1}+x_{1}+l}&\frac{1}{x_{1}+x_{2}+l}\\
			\frac{1}{x_{2}+x_{1}+l}&\frac{1}{x_{2}+x_{2}+l}
		\end{bmatrix}
	\end{eqnarray*} and
	\begin{eqnarray*} 
		A^2&=&\begin{bmatrix}
			\frac{1}{(2x_{1}+l)^2}+\frac{1}{(x_1+x_2+l)^2}&\frac{1}{(2x_1+l)(x_{1}+x_{2}+l)}+\frac{1}{(x_1+x_2+l)(2x_2+l)}\\
			\frac{1}{(2x_1+l)(x_{1}+x_{2}+l)}+\frac{1}{(x_1+x_2+l)(2x_2+l)}&	\frac{1}{(2x_{2}+l)^2}+\frac{1}{(x_1+x_2+l)^2}
		\end{bmatrix}\\
		&=&\begin{bmatrix}
			a_{11}&a_{12}\\
			a_{21}&a_{22}
		\end{bmatrix}.
	\end{eqnarray*}
	Since,
	\begin{eqnarray*}
		a_{12}&=&\frac{1}{(2x_1+l)(x_{1}+x_{2}+l)}+\frac{1}{(x_1+x_2+l)(2x_2+l)}\\
		&=& \frac{1}{(x_1+x_2+l)}\Bigg(\frac{1}{2x_1+l}+\frac{1}{2x_2+l}\Bigg)\\
		&=&\frac{2}{(2x_1+l)(2x_2+l)}.
	\end{eqnarray*}
	Therefore, $a_{12} \in U(GR(p^s,~p^{sm}))$. This implies $A^2\neq I.$ Hence, there is no Cauchy TYPE-II construction, which is almost involutory.
\end{proof}
\section{\bf Construction of new MDS matrices with the help of fixed MDS matrix}
\subsection{MDS automorphisms over Galois ring}

In \cite{33}, Sakalli et al. constructed a new MDS matrix over a finite field of characteristic 2 with the help of automorphisms. 
In this subsection, we investigate  MDS matrices over the extension of Galois rings and distinct functions related to automorphisms. Moreover, one can generate new MDS matrices over the same Galois extension ring by applying these automorphisms and distinct functions to
any MDS matrix.  We begin our discussion with the following result:
\begin{prop}\label{p1}
	Let $R' = \text{GR}(p^s, ~p^{sm})$  be a Galois ring containing $R = \text{GR}(p^s,~ p^{s})$ as a subring.  Let  $A$ be a $k \times k$ matrix over  $R'$. Suppose $A'$ be the matrix generated by applying any distinct automorphism
	\[
	f_{i}: b = a_{0} + a_{1}\xi + a_{2}\xi^2 + \cdots + a_{m-1}\xi^{m-1} \mapsto a_{0} + a_{1}\xi^{p^i} + \cdots + a_{m-1}\xi^{(m-1)p^i}
	\]
	to the elements of $A$, with $0 \leq i \leq m-1$ and $b \in U(R')$, where $R'=R[\xi]$. Then, the determinant of $A'$ is a zero divisor if and only if $\det(A)$ is a zero divisor.
\end{prop}
\begin{proof}
	By   \cite[Theorem 14.30]{11}, the automorphism of $GR(p^s,~p^{sm})$ over $\mathbb{Z}_{p^s}$ are given as 
	$$a_{0}+a_{1}\xi+a_{2}\xi^{2p^m}+\cdots+a_{m-1}\xi^{m-1}\mapsto a_{0}+a_{1}\xi^{p^m}+a_{2}\xi^2+\cdots+a_{m-1}\xi^{(m-1)p^m}.$$
	These mappings are one-to-one because each element in $\mathbb{Z}_{p^s}$ maps to itself. Therefore, the determinant of any matrix obtained by applying an automorphism to $A$ remains unchanged, whether it is a zero divisor or a unit. Specifically, if $\det(A) \in U(GR(p^s, p^{sm}))$ or $\det(A) \in \mathcal{N}(GR(p^s, p^{sm}))$, then $\det(A') \in U(GR(p^s, p^{sm}))$ or $\det(A') \in \mathcal{N}(GR(p^s, p^{sm}))$, respectively.
\end{proof}
\begin{theorem}\label{t7}
	Let $R'=GR(p^s,p^{sm})$ be a Galois ring containing $R=GR(p^s,p^{s})$ as a subring. Then, there exist $ p^{(s-1)m}(p^{m}-1)$ distinct bijective functions related to the automorphisms in the form
	$f_{i,c}:b=(a_{0}+a_{1}\xi+a_{2}\xi^2+\cdots+a_{m-1}\xi^{m-1})\mapsto (a_{0}+a_{1}\xi^{p^m}+\cdots+a_{m-1}\xi^{(m-1)p^m})\cdot c$, where $b$ is any primitive element of $R'$ and $c \in U(R')$  with $0 \leq i \leq m - 1$. These functions
	preserve the MDS property of a square matrix, i.e., new MDS matrices are
	generated from the existing ones.
\end{theorem}
\begin{proof}
	Here we need to show that the properties of being an MDS matrix are satisfied after applying distinct functions with the help of Frobenius
	automorphism. The main idea depends on the fact that every square submatrix of an MDS matrix is non-singular. Note that all elements of an MDS matrix must be unit elements of  Galois ring. We divide proof into three parts. 
	\begin{enumerate}
		\item Letting $f_{i}:~GR(p^s,p^{sm})\rightarrow GR(p^s,~p^{sm})$ defined as $$a_{0}+a_{1}\xi +\cdots + a_{m-1}\xi^{m-1}\mapsto a_{0}+a_{1}\xi^{p^i}+\cdots + a_{m-1}\xi^{(m-1)p^i}.$$
		Then, we have $\det(A')=f_{i}(\det(A))$. If $\det(A)\in U(GR(p^s,p^{sm}))$, then $f_{i}(\det(A))\in U(GR(p^s,~p^{sm}))$, since $f_{i}$ is an automorphism.
		\item Let $g_{c}:GR(p^s,~p^{sm})\rightarrow GR(p^s,~p^{sm})$ defined as $g_{c}(x)=c\cdot x$, where $c\in U(GR(p^s,~p^{sm}))$. Then, $\det(A)'=c\cdot \det(A)$. Since $c\cdot\det(A)\in U(GR(p^s,~p^{sm}))$, so we conclude that 
		$\det(A)\in U(GR(p^s,~p^{sm})).$
		\item Now let $f_{i,c}:GR(p^s,~p^{sm})\rightarrow GR(p^s,~p^{sm})$ defined as the 
		$$f_{i,c}:b=(a_{0}+a_{1}\xi+a_{2}\xi^2+\cdots+a_{m-1}\xi^{m-1})\mapsto (a_{0}+a_{1}\xi^{p^m}+\cdots+a_{m-1}\xi^{(m-1)p^m})\cdot c$$
		Since,
		\begin{eqnarray*}
			f_{i,c}=(g_{c}o~f_{i})(\beta)&=&g_{c}(f_{i}(\beta))\\
			&=&g_{c}(\xi^{p^i})\\
			&=&c\cdot\xi^{p^i}. 
		\end{eqnarray*}
		Then, $\det(A')=c\cdot f_{i}(\det(A))$. Since $\det(A')\in U(GR(p^s,p^{sm}))$, so we have $~\det(A')\in U(GR(p^s,~p^{sm})).$
	\end{enumerate}
	An automorphism $f_{m-i}:GR(p^s,~p^{sm})\rightarrow GR(p^s,~p^{sm})$ defined as 
	$$f_{m-i}(a_{0}+a_{1}\xi+\cdots+a_{m-1}\xi^{m-1})=a_{0}+a_{1}\xi^{p^{m-i}}+a_{2}\xi^{2p^{m-i}}+\cdots+a_{m-1}\xi^{(m-1)p^{m-i}}.$$
	Since $f_{m-i}(\xi)=\xi^{p^{m-i}}$, we have
	\begin{eqnarray*}
		f_{i}of_{m-i}(\xi)=f_{i}(\xi^{p^{m-i}})&=&(\xi^{p^{m-i}})^{p^i}\\
		&=&\xi.
	\end{eqnarray*}
	This implies that $\det(A)= f_{m-i}(\frac{1}{c}\det(A')).$ 
	In conclusion, we get $\det(A')\in U(GR(p^s,~p^{sm}))~iff~$ $\det\in U(GR(p^s,~p^{sm})).$
\end{proof}
Application of Theorem \ref{t7} yield the following corollaries:
\begin{corollary}\cite[Theorem 1]{33}
	There exist $m \cdot (2^m - 1)$ distinct  bijective functions related to the automorphisms of the form
	$f_{i,c}: \beta\mapsto (\beta^{2^i})\cdot c$, where $\beta$ is any primitive element of $F_{2^m}$, $c \in F^*_{2^m}$ , and $0 \leq i \leq m - 1$. These functions
	preserve the MDS property of a square matrix over the same q-array extension field, that is, we can generate new MDS matrices from a given MDS matrix.
\end{corollary}

Moreover, if we take $s=1$ in Theorem \ref{t7}, then our result reduces to Galois field of characteristic $p.$
\begin{corollary}
	Let $A$ be a $k \times k$ matrix over the finite field $F_{p^m}$. Let $A'$
	be generated by applying any distinct
	automorphism $f_{i}:b\mapsto b^{p^i}$ to the elements of $A$ with $0 \leq i \leq m - 1$ and $b \in F^*_{p^m}$. Then, the determinant of $A'$
	is equal to 0 if and only if the determinant of $A$ is equal to 0.
\end{corollary}

\begin{corollary}
	There exist $m \cdot (p^m - 1)$ distinct  bijective functions related to the automorphisms in the form
	of $f_{i,c}: \beta\mapsto (\beta^{p^i})\cdot c$, where $\beta$ is any primitive element of $F_{p^m}$, $c \in F^*_{p^m}$, and $0 \leq i \leq m - 1$. These functions
	preserve the MDS property over the same q-array extension field, i.e., new MDS matrices are
	generated from the existing ones.
\end{corollary}
%
	%
	%
%

\subsection{MDS Isomorphism over Galois ring}

In \cite{33}, Sakalli et al. constructed a new MDS matrix over a finite field of characteristic 2 with the help of isomorphisms. Our work is motivated by the above mentioned study \cite{33} and we use finite rings instead of finite fields. In fact, we utilize isomorphisms within the Galois ring to establish new bijective functions and define new MDS matrices by using these bijective functions. In Proposition \ref{p21} below, we investigate the non-singularity of the matrices using these isomorphisms.

Throughout this subsection, $GR(p^s,~p^{sm})|_{h(x)}$ represents the Galois ring generated by the basic irreducible polynomial $h(x)$.

\begin{prop}\label{p21}
	Let $A$ and $A'$ be two $k\times k$ matrices over the Galois ring $GR(p^s,~p^{sm})$ generated by $(h_1,~\eta_1)$ and $(h_2,~\eta_2)$, respectively, where $h_i(x)\in F_{p^m}[x]$ and $\eta_{i}$ be any element of order $p^m-1$, for $1\leq i \leq 2.$ Consider  isomorphisms defined by
	$$f_{s_u}:\eta_1\mapsto \eta^{s_u}_2,~s_{u}=e\cdot2^i~for~1\leq e\leq 2^m-2,~gcd(e,2^m-1)=1,~h_{1}(\eta^{s_u}_2)=0.$$ Then, determinant of $A'$ is  zero divisor iff determinant of $A$ is zero divisor.
\end{prop}
\begin{proof}
	Let $R = GR(p^s,~p^{sm})$ be a Galois ring of order $p^{sm}$ with characteristic $p^s$ and let $\eta_1,~\eta_2 \in GR(p^s,~p^{sm})$ be roots of the basic irreducible polynomials $h_{1}(x)$ and $h_2(x)$ of degree $m$ over $\mathbb{Z}_{p^s}$, respectively. Then, by \cite[Theorem 14.30]{11}, we have
	
	\begin{eqnarray}\label{eqn12}
		\phi_{1}&:&~\frac{\mathbb{Z}_{p^s}[x]}{(h_1(x))}\rightarrow \mathbb{Z}_{p^s}[\eta_1],\label{eqn51}\\
		\phi_{2}&:&~\frac{\mathbb{Z}_{p^s}[x]}{(h_{2}(x))}\rightarrow \mathbb{Z}_{p^s}[\eta^{s_u}_2],\label{eqn52}\\
		\phi &:&~\frac{Z_{p^s}[x]}{(h_{1}(x))}\rightarrow \frac{Z_{p^s}[x]}{(h_{2}(x))}\label{eqn53}.
	\end{eqnarray}
	Define a map $f_{s_u}$ with the help of Equations (\ref{eqn51}), (\ref{eqn52}),  (\ref{eqn53}):\\
	$$f_{s_u}=\phi_{2}o\phi o\phi_1^{-1}:\mathbb{Z}_{p^s}[\eta_1]\rightarrow \mathbb{Z}_{p^s}[\eta^{s_u}_2],$$
	such that $$f_{s_u}(a_{0}+a_{1}\eta_1+a_{2}\eta^2_2+\cdots+a_{m-1}\eta^{m-1}_1)=a_{0}+a_1\eta^{s_u}_2+\cdots+a_{m-1}\eta^{{s_u}(m-1)}_2$$
	$f_{s_u}(\eta_1)=\eta^{s_u}_2.$ Hence, $f_{s_u}$ is an isomorphism, because of $gcd(e,~p^m-1)=1~and~h_1(\eta^{s_u}_2)=0,~where~s_u=e\cdot p^i,~1\leq e \leq p^m-2.$
\end{proof}

\begin{theorem}\label{t44}
	There exist $ p^{(s-1)m}(p^{m}-1)$ distinct bijective functions obtained by using isomorphism in the form of $f_{s,u}:\eta_1 \mapsto \eta^{s_u}_2\cdot c$, where $\eta_1$ and $\eta_2$ are the root of the polynomials $h_{1}(x)$ and $h_{2}(x)$, respectively, $c\in U(GR(p^s,p^{sm})),~s_{u}=e\cdot 2^i,~1\leq e\leq 2^m-2,~gcd(e,~2^m-1)=1,~p_{1}(\beta^{s_u}_2)=0,$ and $0\leq u,i\leq m-1.$ 
	
\end{theorem}
\begin{proof}
	Proof is similar to that of Theorem \ref{t7}.
\end{proof}

\begin{remark}
	These function can be used in generating new MDS matrices over $GR(p^s,p^{sm})|_{h_2(x)}$ from an MDS matrix over $GR(p^s,p^{sm})|_{h_1(x)}$, which preserves the MDS property of a square matrix.
\end{remark}
Application of Theorem \ref{t44} yield the following corollaries:
\begin{corollary}\cite[Theorem 4]{33}
	There exist $m \cdot (2^m - 1)$ distinct functions obtained by using isomorphisms in the form of $f_{s_u}^c: \beta_1 \mapsto (\beta_{2}^{s_u}) \cdot c$ where $\beta_1$ and $\beta_2$ are  any primitive element of $F_{2^m}/p_1(x)$ and $F_{2^m}/p_2(x)$, respectively, $c \in F^*_{2^m}$, $s_u = e \cdot 2^i$ for $1 \leq e \leq 2^m - 2$, $\gcd(e, p^m - 1) = 1$, $p_1(\beta_{2}^{s_u}) = 0$, and $0 \leq u, i \leq m - 1$. These functions can be used in generating new MDS matrices over $F_{2^m}/p_2(x)$ from an MDS matrix over $F_{2^m}/p_1(x)$ which preserve the MDS property of a square matrix.
\end{corollary}
Now, if we take $s=1$ in Theorem \ref{t44}, then our result reduces to Galois field of characteristic $p.$
\begin{corollary}
	Let $A$ be a $k \times k$ matrix over the finite field $F_{p^m}/p_1(x)$ and $\beta_1$ be any primitive element of $F_{p^m}/p_1(x)$. Let $A'$ be a $k \times k$ matrix over the finite field $F_{p^m}/p_2(x)$ generated by applying the isomorphism $f_{s_u}: \beta_1 \mapsto \beta_{p}^{s_u}$ to the elements of $A$ (which can also be represented as $\beta_1^d$ for $0 \leq d \leq p^m - 2)$ where $\beta_2$ is any primitive element of $F_{p^m}/p_2(x)$, $s_u = e \cdot p^i$ for $1 \leq e \leq p^m - 2$, $\gcd(e, p^m - 1) = 1$, $p_1(\beta_{p}^{s_u}) = 0$, and $0 \leq u, i \leq m - 1$. Then $\det(A')=0$ iff $\det(A)=0.$
\end{corollary}

\begin{proof}
	The proof is similar to Proposition \ref{p1}, since we have the same mapping up to the isomorphism and all entries of an MDS matrix remain nonzero after applying the isomorphism. Note that, each $f_{s_u}$ maps each element in $F_2$ to itself. The isomorphism $f_{s_u}$ is related to automorphism as defined in Proposition \ref{p1} due to the structure of $s_u$.
\end{proof}

\begin{corollary}
	There exist $m \cdot (p^m - 1)$ distinct functions obtained by using isomorphisms in the form of $f_{s_u}^c: \beta_1 \mapsto (\beta_{p}^{s_u}) \cdot c$ where $\beta_1$ and $\beta_2$ are respectively any primitive element of $F_{p^m}/p_1(x)$ and $F_{p^m}/p_2(x)$, $c \in F^*_{2^m}$, $s_u = e \cdot p^i$ for $1 \leq e \leq p^m - 2$, $\gcd(e, p^m - 1) = 1$, $p_1(\beta_{p}^{s_u}) = 0$, and $0 \leq u, i \leq m - 1$. These functions can be used in generating new MDS matrices over $F_{p^m}/p_2(x)$ from an MDS matrix over $F_{p^m}/p_1(x)$ which preserve the MDS property of a square matrix.
\end{corollary}

\begin{proof}
	Let $\beta \in F_{p^m}$ be a primitive element. Recall that the minimal polynomial of the set $\beta, \beta^2, \ldots, \beta^{p^m-1}$ where $m$ is the smallest integer such that $\beta^{p^m} = \beta$ is the same. Since the proof is similar to Theorem \ref{t7}, we omit it.
\end{proof}
In the following results, we prove that involutory property of matrix is preserve under ring automorphism:
\begin{theorem}
	Let $A=(a_{ij}) $ be an involutory matrix of order $n$ over $GR(p^s,p^{sm})$ and $\phi: GR(p^s,p^{sm}) \rightarrow GR(p^s,p^{sm})$ be any automorphism. Then, $A'=(\phi(a_{ij}))$ is  an involutory matrix.
\end{theorem}
\begin{proof}
	Since $A$ is an involutory matrix over $GR(p^s,p^{sm})$, so we have 
	
	\begin{equation}\label{eqn3}
		\sum_{k=1}^{n}a_{ik}a_{kj}=\delta_{ij};~~1\leq i,~j\leq n,~~where
	\end{equation}
	$$\delta_{ij}= 
	\begin{cases}
		1, & i\neq j\\
		0, & i=j\\
	\end{cases}.$$
	Applying $\phi$ on Equation (\ref{eqn3}), we get \\
	$$	\phi\Bigg(\sum_{k=1}^{n}a_{ik}a_{kj}\Bigg)=\sum_{k=1}^{n}\phi(a_{ik})\phi (a_{kj})=\delta_{ij},~~1\leq i,~j\leq n.$$
	Hence, $A'=(\phi(a_{ij}))$ is  an involutory matrix.
\end{proof}
\begin{corollary}
	Let $A=(a_{ij}) $ be an involutory MDS matrix of order $n$ over $GR(p^s,p^{sm})$ and $\phi: GR(p^s,p^{sm}) \rightarrow GR(p^s,p^{sm})$ be any automorphism. Then, $A'=(\phi(a_{ij}))$ is an involutory MDS matrix.
\end{corollary}
%
\section{\textbf{The examples}}



			In this section, we present some examples of Cauchy MDS matrices over Galois rings with characteristics both even and odd. Additionally, we provide some  examples in which number of  entries in the matrix is reduced. Furthermore, we illustrate the construction of new MDS matrices by utilizing a given Cauchy matrix in combination with the Frobenius automorphisms.
			\par We construct examples of Cauchy MDS matrices over the Galois ring $GR(2^2,2^8)$ of order 7 by using Theorem \ref{t6} and Theorem \ref{t8}, respectively.
			\begin{example}\label{exa1}
				Let $GR(2^2,2^8)=\frac{\mathbb{Z}_{4}[x]}{(x^4+2x^2+3x+1)}$ be Galois ring of characteristic 4 with cardinality 256 and $\xi=x+(x^4+2x^2+3x+1)$. Define $\tau= \{0,1,\xi,\xi^2,\dots, \xi^{15}\}$, where $o(\xi)=15$.\\
				For	$x_{i}=\xi^{i}$ and $y_{j}=\xi^{j+7}$, where $0\leq i,j \leq 6$, define Cauchy matrix of the first form,
				\begin{eqnarray*}
					A&=&\begin{bmatrix}
						\frac{1}{1-\xi^7}&\frac{1}{1-\xi^8}&\frac{1}{1-\xi^9}&\frac{1}{1-\xi^{10}}&\frac{1}{1-\xi^{11}}&\frac{1}{1-\xi^{12}}&\frac{1}{1-\xi^{13}}\\
						
						\frac{1}{\xi-\xi^7}&\frac{1}{\xi-\xi^8}&\frac{1}{\xi-\xi^9}&\frac{1}{\xi-\xi^{10}}&\frac{1}{\xi-\xi^{11}}&\frac{1}{\xi-\xi^{12}}&\frac{1}{\xi-\xi^{13}}\\
						\frac{1}{\xi^2-\xi^7}&\frac{1}{\xi^2-\xi^8}&\frac{1}{\xi^2-\xi^9}&\frac{1}{\xi^2-\xi^{10}}&\frac{1}{\xi^2-\xi^{11}}&\frac{1}{\xi^2-\xi^{12}}&\frac{1}{\xi^2-\xi^{13}}\\
						\frac{1}{\xi^3-\xi^7}&\frac{1}{\xi^3-\xi^8}&\frac{1}{\xi^3-\xi^9}&\frac{1}{\xi^3-\xi^{10}}&\frac{1}{\xi^3-\xi^{11}}&\frac{1}{\xi^3-\xi^{12}}&\frac{1}{\xi^3-\xi^{13}}\\
						\frac{1}{\xi^4-\xi^7}&\frac{1}{\xi^4-\xi^8}&\frac{1}{\xi^4-\xi^9}&\frac{1}{\xi^4-\xi^{10}}&\frac{1}{\xi^4-\xi^{11}}&\frac{1}{\xi^4-\xi^{12}}&\frac{1}{\xi^4-\xi^{13}}\\
						\frac{1}{\xi^5-\xi^7}&\frac{1}{\xi^5-\xi^8}&\frac{1}{\xi^5-\xi^9}&\frac{1}{\xi^5-\xi^{10}}&\frac{1}{\xi^5-\xi^{11}}&\frac{1}{\xi^5-\xi^{12}}&\frac{1}{\xi^5-\xi^{13}}\\
						\frac{1}{\xi^6-\xi^7}&\frac{1}{\xi^6-\xi^8}&\frac{1}{\xi^6-\xi^9}&\frac{1}{\xi^6-\xi^{10}}&\frac{1}{\xi^6-\xi^{11}}&\frac{1}{\xi^6-\xi^{12}}&\frac{1}{\xi^6-\xi^{13}}\\

					\end{bmatrix}\\
				\end{eqnarray*}
				\scalebox{0.60}{	$={\begin{bmatrix}
							2\xi^3+3\xi^2&2\xi^3+\xi^2+1&\xi^2+\xi &\xi^3+2\xi^2+1&2\xi^2+3\xi+2&3\xi^3+\xi^2+2&\xi^3+3\xi^2+3\xi\\
							3\xi^3+\xi^2+\xi+3&2\xi^2+3\xi&3\xi^3+3\xi^2+3&\xi+1&2\xi^3+\xi^2+2\xi+3&2\xi^3+2\xi^2+2\xi+3&2\xi^3+\xi^2+\xi\\
							3\xi^3+2\xi^2+2\xi&\xi^3+2\xi^2+3\xi+1&2\xi+3&3\xi^2+3\xi+3&3\xi^3+\xi^2+2\xi+1&\xi^3+\xi^2+3\xi+2&\xi^3+\xi^2+2\\
							3\xi^2+2\xi+1&2\xi^3+\xi^2+3\xi+1&3\xi^3+2\xi^2+3&\xi^3+3\xi^2+2\xi+2&\xi^3+3\xi^2+\xi+3&3\xi^3+3\xi+2&2\xi^3+3\xi^2+\xi+3\\
							3\xi^2+\xi+3&3\xi^3+\xi^2+\xi+2&2\xi^3+2\xi^2+3\xi+2&\xi^3+2\xi^2&2\xi^3+3\xi^2+3\xi+2&\xi^3+\xi+1&2\xi^3+\xi^2+3\\
							\xi^2+\xi+3&\xi^3+\xi^2+\xi&2\xi^2+\xi&\xi^3+2&\xi^2+3\xi+2&\xi^3+2\xi^2+3\xi+1&3\xi^2+2\xi+1\\
							3\xi+1&2\xi^3+3&3\xi^3+2\xi^2+\xi+1&\xi^2+\xi&\xi^3+\xi^2+2\xi+2&\xi+2&\xi^3+\xi^2+2\xi+3

					\end{bmatrix}}.$}\\

				For	$x_{i}=\xi^{i-1}$ and $y_{j}=\xi^{j+6}$, define Cauchy matrix of the second form,
				\begin{eqnarray*}
					A&=&\begin{bmatrix}
						\frac{1}{1+\xi^7}&\frac{1}{1+\xi^8}&\frac{1}{1+\xi^9}&\frac{1}{1+\xi^{10}}&\frac{1}{1+\xi^{11}}&\frac{1}{1+\xi^{12}}&\frac{1}{1+\xi^{13}}\\
						
						\frac{1}{\xi+\xi^7}&\frac{1}{\xi+\xi^8}&\frac{1}{\xi+\xi^9}&\frac{1}{\xi+\xi^{10}}&\frac{1}{\xi+\xi^{11}}&\frac{1}{\xi+\xi^{12}}&\frac{1}{\xi+\xi^{13}}\\
						\frac{1}{\xi^2+\xi^7}&\frac{1}{\xi^2+\xi^8}&\frac{1}{\xi^2+\xi^9}&\frac{1}{\xi^2+\xi^{10}}&\frac{1}{\xi^2+\xi^{11}}&\frac{1}{\xi^2+\xi^{12}}&\frac{1}{\xi^2+\xi^{13}}\\
						\frac{1}{\xi^3+\xi^7}&\frac{1}{\xi^3+\xi^8}&\frac{1}{\xi^3+\xi^9}&\frac{1}{\xi^3+\xi^{10}}&\frac{1}{\xi^3+\xi^{11}}&\frac{1}{\xi^3+\xi^{12}}&\frac{1}{\xi^3+\xi^{13}}\\
						\frac{1}{\xi^4+\xi^7}&\frac{1}{\xi^4+\xi^8}&\frac{1}{\xi^4+\xi^9}&\frac{1}{\xi^4+\xi^{10}}&\frac{1}{\xi^4+\xi^{11}}&\frac{1}{\xi^4+\xi^{12}}&\frac{1}{\xi^4+\xi^{13}}\\
						\frac{1}{\xi^5+\xi^7}&\frac{1}{\xi^5+\xi^8}&\frac{1}{\xi^5+\xi^9}&\frac{1}{\xi^5+\xi^{10}}&\frac{1}{\xi^5+\xi^{11}}&\frac{1}{\xi^5+\xi^{12}}&\frac{1}{\xi^5+\xi^{13}}\\
						\frac{1}{\xi^6+\xi^7}&\frac{1}{\xi^6+\xi^8}&\frac{1}{\xi^6+\xi^9}&\frac{1}{\xi^6+\xi^{10}}&\frac{1}{\xi^6+\xi^{11}}&\frac{1}{\xi^6+\xi^{12}}&\frac{1}{\xi^6+\xi^{13}}\\

					\end{bmatrix}.\\
				\end{eqnarray*}
			\end{example}\label{exa9}
			
			Next example justifies  Theorem \ref{t17}. In this example, we reduce number of entries of  matrix with the help of nilpotent element in Galois ring $GR(3^2,(3^2)^3)=\frac{\mathbb{Z}_{3^2}[x]}{(x^3+3x^2+2x+4)}.$ 
			
			\begin{example}\label{exa2}
				Let $GR(3^2,(3^2)^3)=\frac{\mathbb{Z}_{3^2}[x]}{(x^3+3x^2+2x+4)}$ be Galois ring of characteristic 9 with cardinality 729 and $\xi=x+(x^3+3x^2+2x+4)$ such that $o(\xi)=26$. Define $\tau= \{0,1,\xi,\xi^2,\dots, \xi^{25}\},$ $x_{i}=\xi^i;~~0\leq i\leq 5,$ and $y_{j}=\xi^{5+j};~~1\leq j\leq 6.$
				
				\begin{eqnarray*}
					A&=&\Bigg[\frac{1}{x_{i}+y_{j}}\Bigg]\\
					&=&\begin{bmatrix}
						\frac{1}{1+\xi^6}&\frac{1}{1+\xi^7}&\frac{1}{1+\xi^8}&\frac{1}{1+\xi^9}&\frac{1}{1+\xi^{10}}&\frac{1}{1+\xi^{11}}\\
						\frac{1}{\xi+\xi^6}&\frac{1}{\xi+\xi^7}&\frac{1}{\xi+\xi^8}&\frac{1}{\xi+\xi^9}&\frac{1}{\xi+\xi^{10}}&\frac{1}{\xi+\xi^{11}}\\
						\frac{1}{\xi^2+\xi^6}&\frac{1}{\xi^2+\xi^7}&\frac{1}{\xi^2+\xi^8}&\frac{1}{\xi^2+\xi^9}&\frac{1}{\xi^2+\xi^{10}}&\frac{1}{\xi^2+\xi^{11}}\\
						\frac{1}{\xi^3+\xi^6}&\frac{1}{\xi^3+\xi^7}&\frac{1}{\xi^3+\xi^8}&\frac{1}{\xi^3+\xi^9}&\frac{1}{\xi^3+\xi^{10}}&\frac{1}{\xi^3+\xi^{11}}\\
						\frac{1}{\xi^4+\xi^6}&\frac{1}{\xi^4+\xi^7}&\frac{1}{\xi^4+\xi^8}&\frac{1}{\xi^4+\xi^9}&\frac{1}{\xi^4+\xi^{10}}&\frac{1}{\xi^4+\xi^{11}}\\
						\frac{1}{\xi^5+\xi^6}&\frac{1}{\xi^5+\xi^7}&\frac{1}{\xi^5+\xi^8}&\frac{1}{\xi^5+\xi^9}&\frac{1}{\xi^5+\xi^{10}}&\frac{1}{\xi^5+\xi^{11}}
					\end{bmatrix}.
				\end{eqnarray*}

				For $l=6 \in \mathcal{N}(GR(3^2,(3^2)^3))$ and $x_{i}=\xi^i;~~0\leq i \leq 5$, define
				\begin{eqnarray*}
					B&=&\Bigg[\frac{1}{x_{i}+y_{j}}\Bigg]=\Bigg[\frac{1}{\xi^i+\xi^j+6}\Bigg]\\
					&=&\begin{bmatrix}
						\frac{1}{8}&\frac{1}{7+\xi}&\frac{1}{7+\xi^2}&\frac{1}{7+\xi^3}&\frac{1}{7+\xi^4}&\frac{1}{7+\xi^5}\\
						\frac{1}{7+\xi}&\frac{1}{6+2\xi}&\frac{1}{6+\xi+\xi^2}&\frac{1}{6+\xi+\xi^3}&\frac{1}{6+\xi+\xi^4}&\frac{1}{6+\xi+\xi^5}\\
						\frac{1}{7+\xi^2}&\frac{1}{6+\xi^2+\xi}&\frac{1}{6+\xi^2+\xi^2}&\frac{1}{6+\xi^2+\xi^3}&\frac{1}{6+\xi^2+\xi^4}&\frac{1}{6+\xi^2+\xi^5}\\
						\frac{1}{7+\xi^3}&\frac{1}{6+\xi^3+\xi}&\frac{1}{6+\xi^3+\xi^2}&\frac{1}{6+\xi^3+\xi^3}&\frac{1}{6+\xi^3+\xi^4}&\frac{1}{6+\xi^3+\xi^5}\\
						\frac{1}{7+\xi^4}&\frac{1}{6+\xi^4+\xi}&\frac{1}{6+\xi^4+\xi^2}&\frac{1}{6+\xi^4+\xi^3}&\frac{1}{6+\xi^4+\xi^4}&\frac{1}{6+\xi^4+\xi^5}\\
						\frac{1}{7+\xi^5}&\frac{1}{6+\xi^5+\xi}&\frac{1}{6+\xi^5+\xi^2}&\frac{1}{6+\xi^5+\xi^3}&\frac{1}{6+\xi^5+\xi^4}&\frac{1}{6+\xi^5+\xi^5}
					\end{bmatrix},
				\end{eqnarray*}
			\end{example}
			where $B$ is symmetric MDS Cauchy matrix of TYPE-II of order 6.\\\\
			In the forthcoming example, we  demonstrate, how a new Cauchy MDS matrix can be create from a given Cauchy matrix, in view of Theorem \ref{t10}.
			\begin{example}\label{exa3}
				Let $GR(2^2,2^8)=\frac{\mathbb{Z}_{4}[x]}{(x^4+2x^2+3x+1)}$ be Galois ring of characteristic 4 with cardinality 256 and $\xi=x+(x^4+2x^2+3x+1)$ such that $o(\xi)=15$. Define  automorphisms 
				\begin{eqnarray*}
					\phi^i : GR(2^2,2^8)&\rightarrow& GR(2^2,2^8)~by\\
					\phi^i(a_{0}+a_{1}\xi+a_{2}\xi^2+a_{3}\xi^3)&=&a_{0}+a_{1}\xi^{2^i}+a_{2}\xi^{2\cdot2^i}+a_{3}\xi^{3\cdot2^i},
				\end{eqnarray*}
				where $0\leq i \leq 3$. This gives
				\begin{eqnarray*}
					\phi(a_{0}+a_{1}\xi+a_{2}\xi^2+a_{3}\xi^3)&=&a_{0}+a_{1}\xi^{2}+a_{2}\xi^{4}+a_{3}\xi^{6},\\
					\phi^2(a_{0}+a_{1}\xi+a_{2}\xi^2+a_{3}\xi^3)&=&a_{0}+a_{1}\xi^{4}+a_{2}\xi^{8}+a_{3}\xi^{12},\\
					\phi^3(a_{0}+a_{1}\xi+a_{2}\xi^2+a_{3}\xi^3)&=&a_{0}+a_{1}\xi^{8}+a_{2}\xi+a_{3}\xi^{9}.
				\end{eqnarray*}
				Suppose $x_{1}=1,~x_{2}=\xi,~x_{3}=\xi^2,~x_{4}=\xi^3$, and $y_{1}=\xi^4,~y_{2}=\xi^5,~y_{3}=\xi^6,~y_{4}=\xi^7.$
				In view of  Theorem \ref{t6}, the Cauchy matrix defined by 
				\begin{eqnarray*}
			B&=&\begin{bmatrix}
						\frac{1}{1-\xi^4}&\frac{1}{1-\xi^5}&\frac{1}{1-\xi^6}&\frac{1}{1-\xi^7}\\
						\frac{1}{\xi-\xi^4}&\frac{1}{\xi-\xi^5}&\frac{1}{\xi-\xi^6}&\frac{1}{\xi-\xi^7}\\
						\frac{1}{\xi^2-\xi^4}&\frac{1}{\xi^2-\xi^5}&\frac{1}{\xi^2-\xi^6}&\frac{1}{\xi^2-\xi^7}\\
						\frac{1}{\xi^3-\xi^4}&\frac{1}{\xi^3-\xi^5}&\frac{1}{\xi^3-\xi^6}&\frac{1}{\xi^3-\xi^7}
					\end{bmatrix}\\
					&=&\begin{bmatrix}
						3+2\xi+2\xi^2+3\xi^3&2+\xi+3\xi^2+2\xi^3&3\xi^2+2\xi^3&2+\xi^2+\xi^3\\
						1+2\xi^2&1+3\xi^2+\xi^3&3+3\xi+2\xi^2+2\xi^3&3\xi+2\xi^2\\
						3\xi+3\xi^2+2\xi^3&1+3\xi^3&1+\xi+\xi^2+3\xi^3&2+2\xi^2+\xi^3\\
						3+2\xi+\xi^2+2\xi^3&3+3\xi+2\xi^2&1+2\xi+3\xi^2+3\xi^3&2+3\xi+3\xi^2+3\xi^3
						
					\end{bmatrix}
				\end{eqnarray*}
				is an MDS matrix.
				
				Further, by Theorem \ref{t10}, the matrices defined by
				\begin{eqnarray*}
					A_{1}&=&\Bigg[\frac{1}{\phi(x_{i}-y_{j})}\Bigg]\\
					&=&\begin{bmatrix}
						3+3\xi^2+3\xi^3&3+\xi+\xi^2+2\xi^3&1+3\xi+2\xi^3&3+3\xi+\xi^2+\xi^3\\
						3+2\xi^2&\xi+\xi^2+\xi^3&1+2\xi+\xi^2+2\xi^3&2+2\xi+3\xi^2\\
						1+3\xi+3\xi^2+2\xi^3&3+2\xi+\xi^2+3\xi^3&2+3\xi+3\xi^3&2+3\xi^2+\xi^3\\
						2+\xi+2\xi^2+2\xi^3&1+2\xi+3\xi^2&\xi+\xi^2+3\xi^3&1+\xi+2\xi^2+3\xi^3
						
					\end{bmatrix},	\\
				\end{eqnarray*}
				\begin{eqnarray*}
					A_{2}&=&\Bigg[\frac{1}{\phi^2(x_{i}-y_{j})}\Bigg]\\
					&=&\begin{bmatrix}
						2+\xi+3\xi^2+3\xi^3&2+\xi+3\xi^2+2\xi^3&1+\xi^2+2\xi^3&3\xi+\xi^3\\
						3+2\xi^2&1+3\xi+2\xi^2+\xi^3&\xi+2\xi^2+2\xi^3&3+3\xi\\
						2+3\xi+3\xi^2+2\xi^3&3\xi+\xi^2+3\xi^3&2\xi+3\xi^3&1+\xi+\xi^2+\xi^3\\
						2\xi+3\xi^2+2\xi^3&2+3\xi&1+3\xi+3\xi^3&1+2\xi^2+3\xi^3

					\end{bmatrix},	
				\end{eqnarray*}
				\begin{eqnarray*}
					A_{3}&=&\Bigg[\frac{1}{\phi^3(x_{i}-y_{j})}\Bigg]\\
					&=&\begin{bmatrix}
						1+\xi+3\xi^3&3+3\xi+\xi^2+2\xi^3&\xi+2\xi^2&2+2\xi+2\xi^2+\xi^3\\
						1+2\xi^2&1+2\xi^2+\xi^3&2+2\xi+3\xi^2+2\xi^3&3+3\xi^2\\
						3+3\xi+3\xi^2+2\xi^3&1+3\xi+2\xi^2+3\xi^3&2+2\xi+3\xi^2+3\xi^3&2+3\xi+2\xi^2+\xi^3\\
						1+3\xi+2\xi^2+2\xi^3&2+3\xi^2&3+2\xi+3\xi^3&1+2\xi^2+3\xi^3
					\end{bmatrix}	
				\end{eqnarray*}
				are three Cauchy MDS matrices of order 4.
			\end{example}
			\begin{remark}
				Note that, by using Theorem \ref{t7} , we can obtain 240 and 702 new MDS matrices in Example \ref{exa1} and \ref{exa2}, respectively.
			\end{remark}
			\begin{example}
				Let $GR(3^2,~3^4)$ be a Galois ring defined by the basic irreducible polynomial $p_1(x)=5x^2+2x+4$. Let $\eta_1=x+(5x^2+2x+4)$ is a primitive root of $p_{1}(x)$ and  $$B=\begin{bmatrix}
					4\xi+1&5&7\xi+3\\
					\xi+2&7\xi+3&5\xi+2\\
					3\xi+2&2\xi&3\xi+1
				\end{bmatrix},$$
				is a 3$\times$3 MDS matrix over $GR(3^2,~3^4)|_{p_1(x)}$, where $GR(3^2,~3^4)|_{p_1(x)}$ represents Galois ring defined by the polynomial $p_1(x)$.  Suppose the primitive element $\eta_2$ of $GR(3^2,~3^4)|_{p_1(x)}$, which is also a root of $5x^2+7x+4.$ Then, one can obtain 2 distinct isomorphisms from $GR(3^2,~3^4)|_{p_1(x)}$ to $GR(3^2,~3^4)|_{p_1(x)}$ by computing $s_u$ values (which are $s_0=5$ and $s_1=7$). These isomorphism are $f_{5,1}:\eta_1\mapsto \eta^5_2$ and $f_{7,1}:\eta_1\mapsto \eta^5_2$. For example, by using isomorphism, we can generate 3$\times$3 MDS matrix $B'= \begin{bmatrix}
					4\xi+1&5&7\xi+3\\
					\xi+2&7\xi+3&5\xi+2\\
					3\xi+2&2\xi&3\xi+1
				\end{bmatrix}$ over $GR(3^2,~3^4)|_{p_1(x)}$ from $B$ over $GR(3^2,~3^4)|_{p_1(x)}$. Similarly, we get another MDS matrix $B''= \begin{bmatrix}
					4\xi+1&5&7\xi+3\\
					\xi+2&7\xi+3&5\xi+2\\
					3\xi+2&2\xi&3\xi+1
				\end{bmatrix}$ over $GR(3^2,~3^4)|_{p_1(x)}$   of order 3. Hence, with the help of Theorem \ref{t44}, we get 72 more MDS matrices over $GR(3^2,~3^4)|_{p_1(x)}.$
				
			\end{example}
					%
					%
					%
					%
					%
					%
					%
					%
					%
					%
					%
					%


\section{\textbf{Conclusion}}
In the present paper,  we constructed MDS matrices using Cauchy matrices over Galois rings. We  developed a novel approach by deploying the Frobenius automorphism within Galois rings. Moreover, we   constructed a Cauchy MDS matrix of order 6 which was not known before using the concept of nilpotent elements. Additionally, we achieved a reduction in the size of Cauchy matrices with the help of nilpotent elements, which offer computational advantages over unit elements due to their eventual reduction to zero after a specific number of operations. Furthermore, we derived functions related to the automorphisms and isomorphisms of Galois ring which preserve MDS property. As a result, these functions can be used to generate new MDS matrices from a given MDS matrix. These are generic functions and can be used for the constructions of other types of MDS matrices, not necessarily Cauchy. In future work, we will try to investigate the conditions under which we can identify compact Cauchy matrices that allow the reduction of the number of entries from \(k^2\) to \(k\).

\section{\bf Declarations}

\noindent \textbf{Funding}\newline
Not applicable.\newline

\noindent \textbf{Data Availability Statement}\newline
Data sharing is not applicable to this article as no data sets were
generated or analyzed during the current study.\newline

\noindent \textbf{Conflicts of Interest}\newline
The authors declare that they have no conflicts of interest. \newline

\end{document}